\documentclass[journal,onecolumn]{IEEEtran}
\ifCLASSINFOpdf
\else
\fi

\usepackage{amsmath,amsfonts}
\usepackage{array}

\usepackage{textcomp}
\usepackage{stfloats}
\usepackage{url}
\usepackage{verbatim}
\usepackage{graphicx}
\usepackage{cite}
\usepackage{subfigure}
\usepackage[linesnumbered,ruled,vlined]{algorithm2e}
\usepackage{xcolor}
\usepackage{amsmath}
\newtheorem{lemma}{Lemma}
\newtheorem{proof}{Proof}
\begin{document}
\title{All-pay Auction Based Profit Maximization in End-to-End Computation Offloading System}

\author{ Hai Xue, \IEEEmembership{Member IEEE}, Yun Xia, Di Zhang, \IEEEmembership{Senior Member, IEEE}, Honghua Wei, Xiaolong Xu, \IEEEmembership{Senior Member, IEEE}
\thanks{This work was supported in part by the National Natural Science Foundation of China under grant (no. 92267104), and the Natural Science Foundation of Jiangsu Province of China under grant (No. BK20211284). \emph{(Corresponding author: Xiaolong Xu.)}

Hai Xue, Yun Xia, and Honghua Wei are with School of Optical-Electrical and Computer Engineering, University of Shanghai for Science and Technology, Shanghai 200093, China (e-mail: hxue@usst.edu.cn, \{xiadayun99, whh2522175060\}@gmail.com). 

Di Zhang is with School of Electrical and Information Engineering, Zhengzhou University, Zhengzhou 450001, China; 
School of Electrical Engineering, Korea University, Seoul 02841, South Korea (e-mail: dr.di.zhang@ieee.org).

Xiaolong Xu is with School of Software, Nanjing University of Information Science and Technology, Nanjing 210044, China (e-mail: xlxu@ieee.org).
}}

\maketitle

\begin{abstract}
Pricing is an important issue in mobile edge computing. How to appropriately determine  the bid of end user (EU) is an incentive factor for  edge cloud (EC) to offer  service. In this letter, we propose an equilibrium pricing scheme based on the all-pay auction model in end-to-end collaboration environment, wherein all EUs can acquire the service at a lower price than the own value of the required  resource. In addition, we propose a set allocation algorithm to divide all the bidders into different sets according to the price, and the EUs in each set  get the service, which averts the case of  getting no service due to the low price. Extensive simulation results demonstrate that the proposed scheme can effectively maximize the total profit of the edge offloading system, and guarantee all EUs can access the service.
\end{abstract}


\IEEEpeerreviewmaketitle

\section{Introduction}

\IEEEPARstart {W} {ith} the propagation of Internet of Things technologies, a large amount of tasks are generated by computation-intensive applications (e.g., virtual reality, video games, and gesture recognition) {\cite{HK}}, {\cite{XC}}, which  consumes substantial computing resources and these applications even exceeds the  computing capacity of edge devices. In order to solve this problem, edge computing is proposed, wherein computation-intensive tasks can be offloaded to the edge servers (a.k.a. end-to-edge collaboration) as well as to the vicinal resource-rich devices (a.k.a. end-to-end collaboration).   For end-to-end collaboration, due to the consumption of computing and communication resources, edge device is not  willing to actively execute tasks from other edge devices if no compensation is  provided.

Resource pricing is proposed as a promising way to tackle the issue that edge device is unwilling to execute tasks from other end users (EUs). Therefore, several recent researches have investigated pricing strategies through game theory to incentivize edge devices to provide services. Wang $et$ $al.$ {\cite{WQ}} proposed a profit maximization multi-round auction (PMMRA) scheme, which  matches buyers and sellers firstly, and then obtains bids through Vickrey auctions.   Liu $et$ $al.$  {\cite{M}}  utilized Stackelberg game  theory to model the interaction between edge clouds (ECs) and  EUs, and EC sets the price according to limited computing power to maximize its revenue. In addition to obtaining the price through game theory, traditional auction models are also utilized to decide the price. For instance, Balzer $et$ $al.$ {\cite{GQ}} investigated the dynamic bidding mechanism of the Dutch auction, and demonstrated that the Dutch auction is the most profitable one among pricing strategies.

 However, most of the existing   studies  focus on  balancing the interests of buyer and seller, ignoring that the server is unwilling to offer service without a satisfactory revenue. Furthermore,  the winner gets the service at a lower cost and the EUs who fail to bid are not punished, which results in a large number of EUs bids causing server overload while executing the matching process between EUs and servers.   Conventional auction schemes also have some problems. For example, for First-Price Sealed-Bid Auction, bidders are reluctant to submit high bids for fear that they may be asked to pay more than their actual valuation, which results in lower profits for ECs and even reluctance to provide services. 

 Taking into account the drawbacks of existing schemes,  this letter  proposes an equilibrium pricing scheme based on all-pay auction model. Different from the conventional auction mechanisms that only the highest bidder gets the bidding item and pays the bids {\cite{JF}},  it is mandatory to pay the quoted price  for all bidders. In order to  guarantee the enthusiasm of  EC, the reservation value is set as the benchmark of the auction value  which is bigger than the value of the EC resource itself. At the  meantime,  the equilibrium bid  is set to maximize the stimulation of EUs to actively offload tasks,  which is much smaller than the value of the EC resource itself  and attracts  risk lovers to participate in the auction. It is beneficial to avoid the possibility that they cannot pay for the value of the EC itself to  obtain the service  for risk lovers.  Simulation results demonstrate that the scheme can effectively enhance the profit of the edge offloading  system.

\section{System model And Problem Formulation}
 In this section, we first illustrate the system model of the proposed scheme. Subsequently, the derivation calculations of equilibrium bid of the EU as well as the reservation value of the ECs are presented. At last, the minimum bid of EUs was determined to maximize the  system profit. 
\subsection{System Model}Fig. \ref{sm} depicts the system model of the proposed scheme, wherein the symmetric independent private value (SIPV) model is adopted {\cite{SS}}. And  the system consists of three layers: EC layer, aggregation layer, and EU layer. Firstly, devices (e.g., smartphones, lab servers, PCs) that have computing power and are temporarily idle are assumed as EC. Secondly, the aggregation layer (e.g., base station,  access point) is the communication bridge between EUs and ECs. Thirdly, the EU (i.e., task owner) layer consists of different devices,\footnote{Unless stated otherwise, we use the term EC and task executor, EU and task owner, interchangeably.} such as sensors, smartphones, or PCs with diverse computational requirements.

In this letter, idle resources of mobile devices are denoted as the available  resource of edge clouds. Due to the limited energy and computing resources of task executors,   we assume that  a task executor can only  process one task in  a time unit.  Denote that each EU state \textit{C} (\textit{v}, \textit{b}, \textit{r}, \textit{m}), \textit{v} is  the valuation  based on its required  resources, \textit{b} is the equilibrium bid  obtained by its own valuation distribution function of the required resources, \textit{r} is the EU reservation value obtained by  \textit{v}. In the following, we take the average reservation value of all EUs as the reservation value of the EC. \textit{m} is a boolean function, where \textit{m} = 0 represents that EU is not served, otherwise, it is served. 

At the beginning of bidding, all EUs give the estimated \textit{v}, and  the aggregation layer device sorts all estimations after receiving them. Here, we introduce the notion of sets. That is, we assign all EUs to different sets \textit{$S_i$} (\textit{i} = 1, 2,..., \textit{k}) according to the valuation size,  the difference between the maximum value $b_{max}$ and the minimum value $b_{min}$ of the valuation in each set cannot be greater than $\epsilon$. To this end, the following equations are obtained:

\begin{equation}
     b_{max}^{S_i} - b_{min}^{S_i} < \epsilon.
\end{equation}

\begin{equation}
    \epsilon = \frac{b-a}{k}.
\end{equation}
where \textit{a} and \textit{b} are the  minimum and  maximum values among all valuations, and \textit{k} is the number of task executors, which is equal to the number of sets. Here, a null set is allowed. 

Subsequently, the aggregation layer device informs the EU with   the  \textit{n}   bidders in its set. Then,  the EU calculates the equilibrium bid \textit{b} and the reservation value \textit{v} according to the  \textit{n}  bidders and the weight number $\lambda$.\footnote{Since the number of bidders \textit{n} will serve as a reference value for the bid of the EU, we set a weight value $\lambda$ to characterize this influence.} After that, it submits the results to the aggregation layer device. The aggregation layer device calculates the average value as the reservation value of EC upon receiving the results. When the sum of the bids of all bidders in the set is greater than the reserved value, the task executor selects the service with the highest bid, and  $m$  is set to 1 accordingly. Otherwise, the task executor does not provide service. That is, the following equation is satisfied: 
\begin{equation}
    \sum_{j=0}^{n_i} b_j \geq \frac{\sum_{j=0}^{n_i} v_j}{n_i}.
    \label{eq3}
\end{equation}
\quad Here, Eq.(\ref{eq3}) represents that the sum of all EU bids must be greater than or equal to the reserved value, where $n_i$ is the number of EUs in set $S_i$.

The aggregation layer device ranks each task executor according to its computing resource size.  While processing the matching  between ECs and task owners,  the task  winner  with the largest retention value is assigned to the task executor with the largest computational resources,   the rest EU-ECs matching can also be deduced in the same manner.
\begin{figure}[tp]
\centering
\includegraphics[width=2.5in]{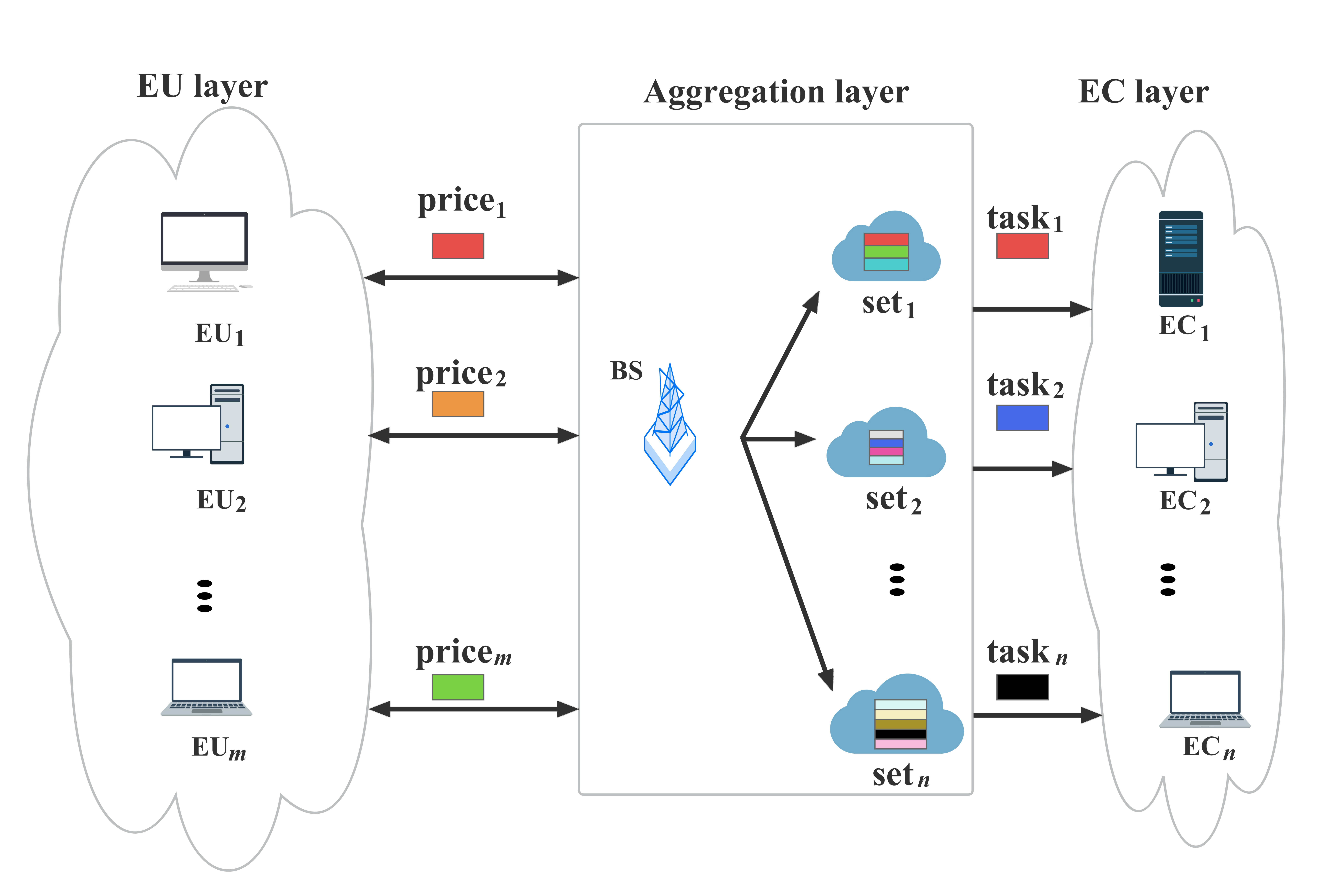}
\caption{System model.}
\label{sm}
\end{figure}

\subsection{ Equilibrium bidding of EUs}
 We define the private valuation of bidder \textit{$i$} is $v_i$, and  $v_i$ $\in$ [$v_{min}^i$, $v_{max}^i$], where $v_{min}^i$ and $v_{max}^i$ are the  lower and upper bounds of the valuation.  It intends to maximize its expected payoff, which is named as the equilibrium bid function. In general, it is assumed to be a strictly increasing function{\cite{AA}},  which implies that  the higher the  private valuation of the bidder, the higher its equilibrium offers. Denote the distribution and density functions of the private valuation \textit{$v$} as \textit{$F(t)$}, \textit{$f(t)$}, respectively. In the  all-pay auction, we  assume that for bidder \textit{i}, whose bid price   to win the auction is given as follows:

\begin{equation}
\begin{aligned}
\label{deqn_ex1a}
 P(t)&=P(v_1<t, v_2<t, \cdots, v_{n-1}<t, v_n<t)\\
 &=F^{n-1}(t).    
\end{aligned}
\end{equation}

 Bidder \textit{i}  pays the offer even if he lost the bid, so the expected revenue  is shown as follows:
\begin{equation}
\label{deqn_ex1b}
 u_{\text {all }}=v F^{n-1}(t)-(1-\lambda) b_{\text {all }}(t)-\frac{n-1}{n} \lambda b_{\text {all }}(t).
\end{equation}
where  $\lambda$  is the correlation coefficient with other EUs, and  $\lambda$ $\in$ [0, 1]. When  \textit{$t$} =  \textit{$v$}, the expected revenue is maximized (i.e., the first derivative is zero):

\begin{equation}
\left.\frac{\mathrm{d} u_{\mathrm{all}}}{\mathrm{d} t}\right|_{t=v}=0.
\end{equation}
\begin{equation}
b_{\text {all }}^{\prime}(v)=v\left[F^{n-1}(v)\right]^{\prime}.
\label{eq7}
\end{equation}
\quad When the bidder has the lowest private valuation,  it is mandatory to withdraw from the auction, that is, the offer price is 0.  Then,  combine Eq.(\ref{eq7}) with Eq.(\ref{eq3}), the following equations are obtained.
\begin{equation}
\left\{\begin{array}{lcl}
b_{\text {all }}^{\prime}(v)=v\left[F^{n-1}(v)\right]^{\prime} \\
b_{\text {all }}(v_{min}^i)=0
\end{array}\right..
\end{equation}
\quad  To sum up, the bidding equilibrium offer function is obtained as follows:
\begin{equation}
\begin{aligned}
    \mathrm{b}_{\mathrm{all}}(\mathrm{t})=\frac{\int_{v_{\min }^i}^v t d F^{n-1}(t)}{1-\frac{\lambda}{n}},\\
    \text{s.t.} \quad v \in [v_{min}^i, v_{max}^i]. 
\end{aligned}
\label{eq9}
\end{equation}
\subsection{Optimal reservation value of ECs}
In order to guarantee the  ECs benefits from the auction and arouses its enthusiasm, the proposed scheme sets a reservation value  to ensure that the ECs will not provide service when the EU bid is too low for market fairness.
For  all-pay auction, if the  reservation value of the tenderer is \textit{r}, we denote the bid of bidder \textit{i} as $b_{all}$(\textit{t}), its expected revenue $u_{all}$  is shown as follows:
\begin{equation}
u_{a l l}=(1-F(r))(\nu F^{n-1}(t)-b_{a l l}(t)).
\label{eq10}
\end{equation}
\quad When \textit{t} = \textit{v},   $u_{all}$  is maximized. Therefore, the first-order necessary conditions  are illustrated as follows:
\begin{equation}
\left.\frac{\mathbf{d}u_{all}}{\mathbf{d}t}\right|_{t=v}=0.
\end{equation}
\begin{equation}
(1-F(r))(v\begin{bmatrix}F^{n-1}(v)\end{bmatrix}'-b_{all}'(v))=0.
\end{equation}
\quad While  the private valuation of bidder \textit{i} is not higher than the reserved value,  it is mandatory to exit the auction, then  $b_{all}$(\textit{t}) = 0, and in conjunction with Eq.(\ref{eq10}), we  can obtain the following formulas.
\begin{equation}
\begin{cases}(1-F(r))(v\Big[F^{n-1}(v)\Big]^{\prime}-b_{all}^{\prime}(v))=0\\ b_{all}(r)=0\end{cases}.
\end{equation}

If the tenderer sets the reservation value as \textit{r}, the expected income based on  \textit{n} bidders by  the all-pay auction is shown as follows:

\begin{equation}
\begin{aligned}U_{a l l}&=n E(b_{all}(v))\\\\ &=n\int_{r}^{v}b_{a l}(v)f(v)\mathrm{d}v\\\\ &=\frac{n}{n-\lambda}\int_{r}^{v}n(n-1)(1-F(t))t F^{n-2}(t)f(t)\mathrm{d}t.\end{aligned}
\end{equation}

If  the private  resource valuation of the tenderer is \textit{$v_0$}, then, the expected surplus of the tenderer  with \textit{n} 
 bidders by  the  all-pay auction are shown as follows:
\begin{equation}
\begin{aligned}R_{a l l}&=U_{a l l}+v_0F^{n}(r)\\ &=\frac{n}{n-\lambda}\int_{r}^{v}n(n-1)(1-F(t))t F^{n-2}(t)f(t)\mathrm{d}t
\\&+v_0F^{n}(r).\end{aligned}
\end{equation}

Assume the optimal reservation value is \textit{r*}, then, the first-order necessary condition   is shown as follows:

\begin{equation}
\begin{array}{c}\frac{\mathrm{d}R_{all}}{\mathrm{d}r}\big|_{r=r^{\ast}}=0.\\ \end{array}
\label{eq16}
\end{equation}
\quad When the number of bidders is \textit{n},  the optimal reservation value \textit{$r^*$} can be obtained to satisfy the following expression by calculating Eq.(\ref{eq16})  with the SIPV auction model:
\begin{equation}
    v_{0}F(r^{\ast})=\frac{n}{n-\lambda}(n-1)(1-F(r^{\ast}))r^{\ast}.
\end{equation}

Since each EU in the set  calculates a reservation value \textit{$r^*$}, and \textit{$R^*$}  is defined as the average  reservation value of all EUs in the set, which is expressed as follows: 
\begin{equation}
    R^*=\frac{\sum_{n=1}^{n_i} r^*_n}{n_i}.
    \label{eq18}
\end{equation}

\subsection{Decision on  minimum valuation}
In the quest to bid and reserve value, it is found that when valuations are too low, there are low equilibrium bids and high reservation values.
To this end, we propose to set the  minimum valuation.\footnote{When it is below the minimum valuation, which implies that EUs can handle the task by itself, without offloading to ECs, and thus no bidding.} For the sake of expression simplicity, we denote $F(x)$ of bidder \textit{i} as the uniform distribution function over [0, \textit{A}]: 
 \begin{equation}
      F (x) =\frac{x}{A}.
 \end{equation}

According to Eq.(\ref{eq9}) and Eq.(\ref{eq16}), the equilibrium bid function and the optimal reserve value function can be transformed into the following  formulas, respectively.
\begin{equation}
    \text{$b_{all}$=}\frac{n-\lambda}{n}\int_0^v t d(\frac{t}{A})^n=\frac{(n-\lambda)(n-1)v_0^n}{n^2A^{n-1}}.
\end{equation}
\begin{equation}
    r^*=A-\frac{v_0(n-\lambda)}{n(n-1)}.
\end{equation}
\quad However, to ensure the ECs offering service, we set $v_i$ $\in$ [$v_0$, \textit{A}], when the valuation is $v_0$,  the reservation value is no more than \textit{n} times the equilibrium bid. Therefore, we obtain the following equation.

\begin{equation}
    \label{eq22}
    A-\frac{v_0(n-\lambda)}{n(n-1)} \leq \frac{(n-\lambda)(n-1)v_0^n}{nA^{n-1}}.
\end{equation}
\quad Note that the critical value of \textit{$v_0$}  is the minimum valuation that the EU should satisfy.

\subsection{Allocation of sets to EUs}
The total number of bidding EUs \textit{N} is divided into \textit{k} sets by  algorithm \ref{alg:alg1}   (see line \ref{step1} to \ref{step16} in algorithm \ref{alg:alg1}). It should be noted that when allocating the set at the beginning, the bidder cannot improve the probability of acquiring service by underestimating and overbidding.  The reason is that if the valuation is too low, it will be assigned to the set with a lower valuation, but the corresponding EC  resources will  be too low  to complete the  task. Thus, the bidder valuation  forces the EU to bid honestly. 
 In each set, if the sum of  all EU bids in the set is not larger than the reserved value of the set, all EUs in the set will not be served (see line \ref{step17} to \ref{step19} in algorithm \ref{alg:alg1}).  That is, the total profit of the system will not be generated, and it is defined as
\begin{equation}
\begin{aligned}\text{W}=\sum_{n=1}^{N'} b_n-\sum_{n=1}^S R^*_n\end{aligned}.
\end{equation}
where \textit{$N'$}  is the total number of bidders in the set being served, \textit{S}  is the number of served sets, \textit{r}  is the reserved value in each set.

\begin{algorithm}[!ht]
    \SetAlgoLined 
	\caption{Set allocation}\label{alg:alg1}
	\KwIn{The estimation \textit{v} $\in$ [$v_0$, \textit{A}] of each bidder, the price \textit{$b_i$}, the reservation value \textit{$r_i$}, the number \textit{k} of sets \textit{S},  the total number of bidding EUs \textit{N}.}
	\KwOut{The highest $b_i$ in each set.}
	Sort all the prices  in ascending order, and put them into array \textit{p}[\textit{n}]\label{step1}\; 
     Put \textit{p}[0] in the first set \textit{$s_1$}\;

	\For {\textit{j} = 0, \textit{j} 
 $\leq$ \textit{n}, \textit{j}++}{
        \For {\textit{i} = \textit{z}+1, \textit{i} 
 $\leq$ \textit{\textit{n}}, \textit{i}++}{
		\eIf{ \textit{p}[\textit{i}] - \textit{p}[\textit{z}] $\leq$ $\epsilon$}{
			Put \textit{p}[\textit{i}] in $s_j$\;
		}{
			\eIf{$s_j$.length \textless  3 (regulation)}{
                Delete the set $s_i$ and its elements\;
            }{
                \textit{z} = \textit{z}+1\;
            }
		}
    Put  $b_i$  into the corresponding set\;
	}
 }\label{step16}
    Calculate each set reservation value $R^*$ according to Eq.(\ref{eq18})\label{step17}\;
    Calculate $b_{sum}$ = $\sum_{i=1}^{n}$$b_i$\;
    If $b_{sum}$ $\geq$ $R^*$, the largest bidder is selected to be served. Otherwise, it will not be served\label{step19}\;
    Set $m_i$ = 1 for all served EUs.
    \label{alg1}
\end{algorithm}

\begin{figure*}[!htbp]
\centering
\subfigure[Effect  of different $\lambda$ on bid.]{
\begin{minipage}[t]{0.33\linewidth}
\centering
\includegraphics[width=2.5in]{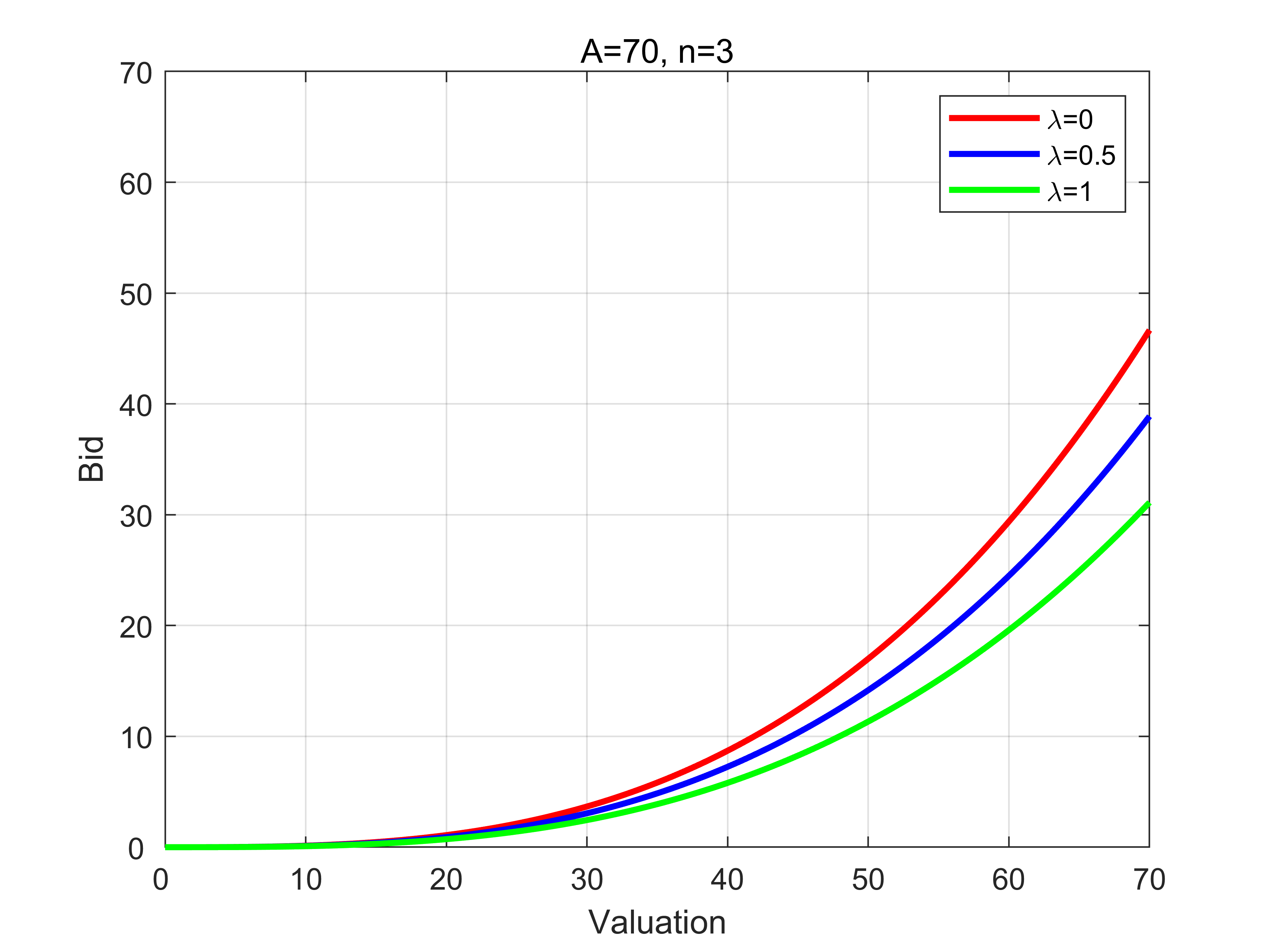}
\label{fig_2}
\end{minipage}%
}%
\subfigure[Effect of different \textit{$n$} on bid.]{
\begin{minipage}[t]{0.33\linewidth}
\centering
\includegraphics[width=2.5in]{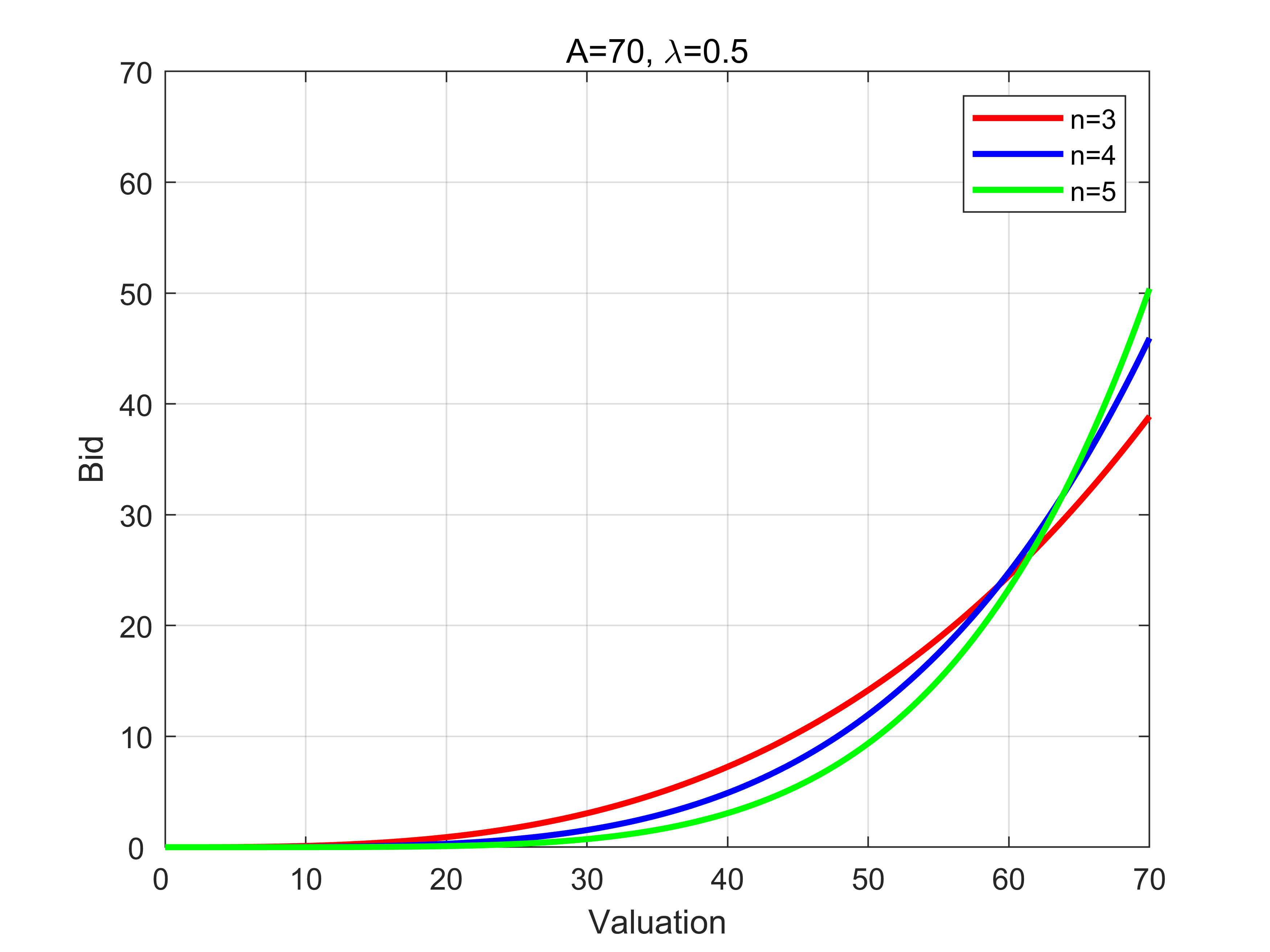}
\label{fig_3}
\end{minipage}%
}%
\subfigure[Effect of different \textit{A} on bid.]{
\begin{minipage}[t]{0.33\linewidth}
\centering
\includegraphics[width=2.5in]{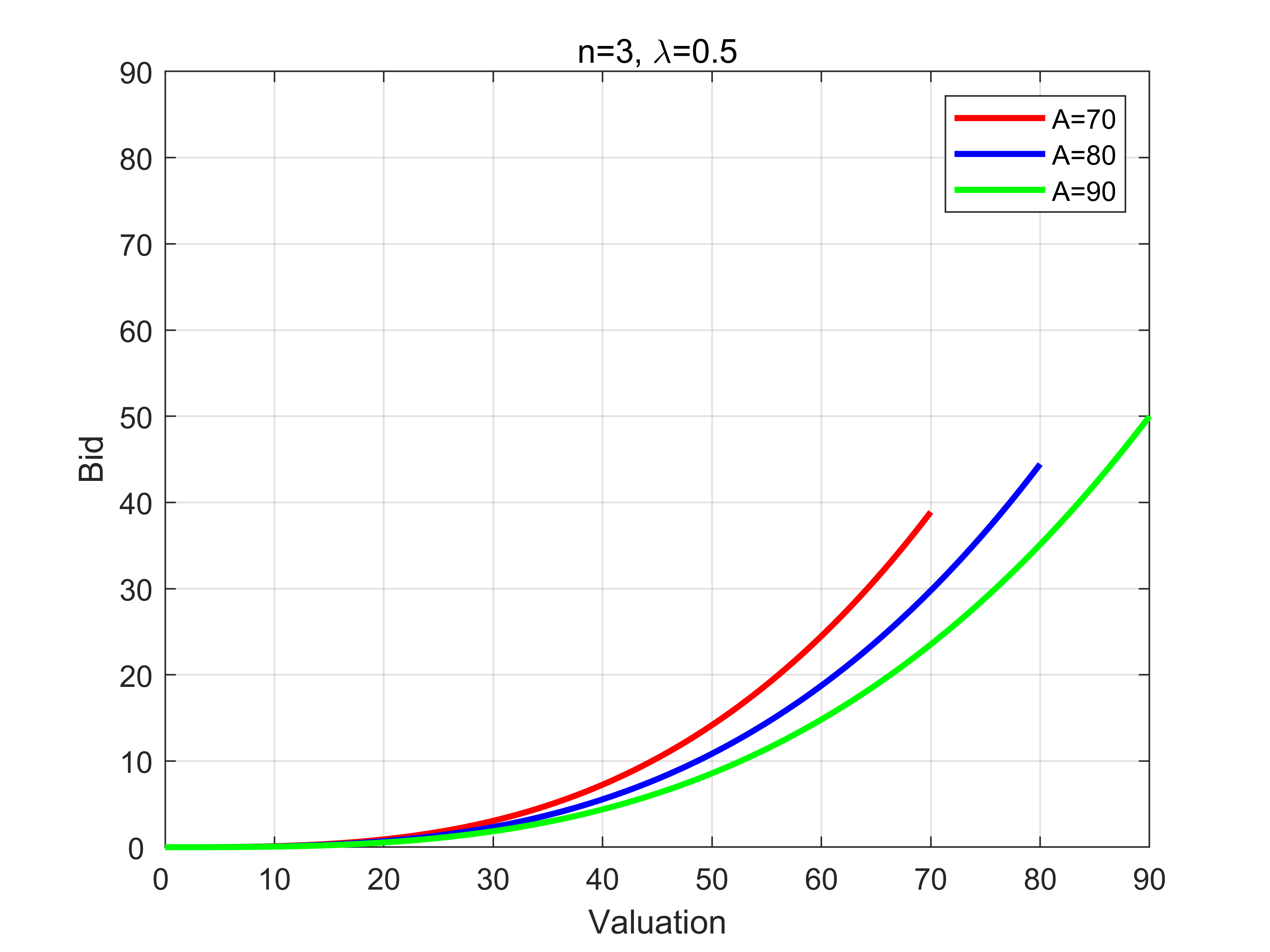}
\label{fig_4}
\end{minipage}
}%
\centering
\caption{Effect of different metric parameters on bid.}
\label{Fig_2}
\end{figure*}

\section{Performance evaluation}
For performance evaluation, we first evaluate different metric parameters on bid. Then, we compare the proposed scheme with 3 existing schemes in terms of the system profit: 1) Greedy algorithm, where all ECs  choose the maximum bid if their resources can serve the EU.  2) Stackelberg \cite{M}, where  EU chooses its bid according to the bid of EC. 3) PMMRA \cite{WQ}, where  EU pays the second highest price to the EC as the final payment. In addition, we set the value of uniform distribution function \textit{A} of all EUs to be 70, 80, 90 randomly, and the correlation coefficient $\lambda$ is set to 0, 0.5, 1. The number of ECs \textit{k} is set to 3, and the number of resources is quantified as 70, 80, 90 {\cite{UG}}.   For each EU, its valuation is between the respective \textit{$v_0$} and \textit{A}. 
\subsection{Bids in different situations}
Fig. \ref{Fig_2} exhibits the equilibrium bids with different parameters of the same valuation function and the equilibrium bids with the same parameters of different valuation functions.

Fig. \ref{fig_2} exhibits the results while \textit{$n$} = 3, \textit{A} = 70.    We can see that with the same valuation, the bids are getting higher as $\lambda$ increases.  This bidding trend can be directly reflected in Eq.(\ref{eq9}). Fig. \ref{fig_3} exhibits the results while \textit{$\lambda$} = 0.5, \textit{A} = 70. It is observed that when the valuation is low, the bids are lower as the number of EUs increases. When valuations are high, they instead lead to higher bids as EUs increase. For EUs, when the valuation is low   and the number of EUs is small, the winning rate is high when they bid high. On the contrary, when the valuation is high  and the number of people is large, the probability of winning is low. To this end, EUs  increase their bids to increase the winning rate. Fig. \ref{fig_4} exhibits the results while \textit{$n$} = 3, $\lambda$ = 0.5. We can see that EUs with smaller \textit{A} bid more for the same valuation they  take. The reason is that as the valuation goes higher, the EU with smaller \textit{A} gets closer to its bid ceiling, beyond which it will lose competitiveness. So it increases its bid to get the service.



\subsection{The comparison of total profit with existing schemes}
 Next, we compare the proposed scheme on total system profit and the bid of winner with the following schemes: PMMRA{\cite{WQ}}, Stackelberg{\cite{M}}, and Greedy algorithm. Note that we have sorted the EUs by valuation size, so the greedy algorithm always achieves the global optimum. Denote  \textit{N} = 12, and $\lambda$ = 0.5. Fig. \ref{fig_5} shows the comparison of total profit  with existing schemes. It is observed that the proposed scheme outperforms the other three existing schemes. However, the effect is slightly worse at the second point.  This is because in the second experiment,  the number of users in each set is the same, the user competitiveness is reduced and the bid is overall lower (refer to Lemma \ref{lemma}).  So the total social profit of this experiment is relatively low, but the gap is not large.

\begin{lemma}
    \label{lemma}
      The lowest total social profit is generated when all EUs are equally divided into each set.
\end{lemma}
\begin{proof}
Please see the Appendix.
\end{proof}

\begin{figure}[!ht]
\centering
\includegraphics[width=2.5in]{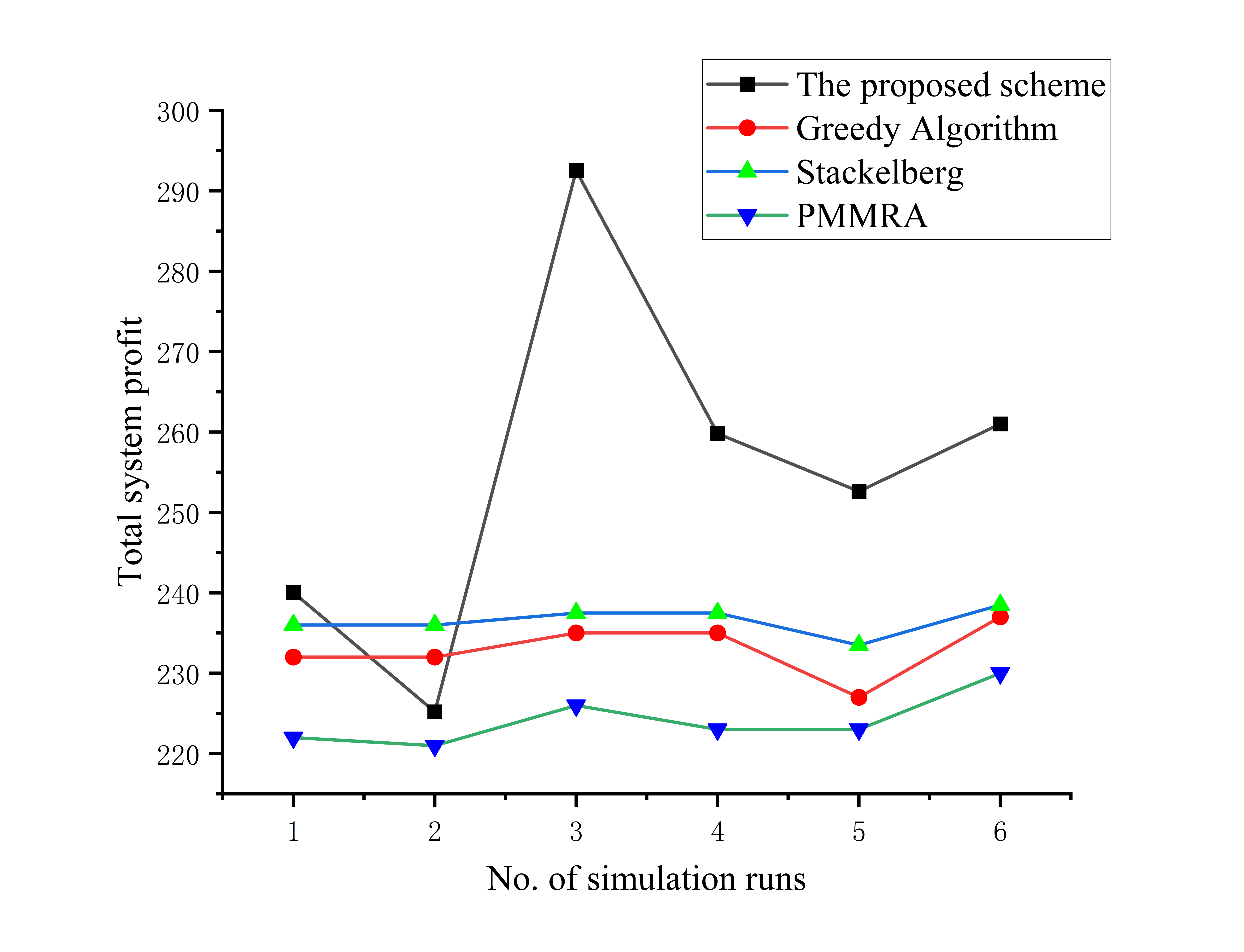}
\caption{Comparison of total profit with existing schemes.}
\label{fig_5}
\end{figure}


\subsection{The comparison of EU bids with existing schemes}
Fig. \ref{fig_6} shows the winner bids of different schemes. We can see that the winner bid of the proposed scheme is significantly smaller than that of the other schemes. When the total social profit is large, the proposed scheme effectively stimulates EUs to participate in the auction, so that the winner can obtain greater benefits with a smaller bid.
\begin{figure}[!ht]
\centering
\includegraphics[width=2.5in]{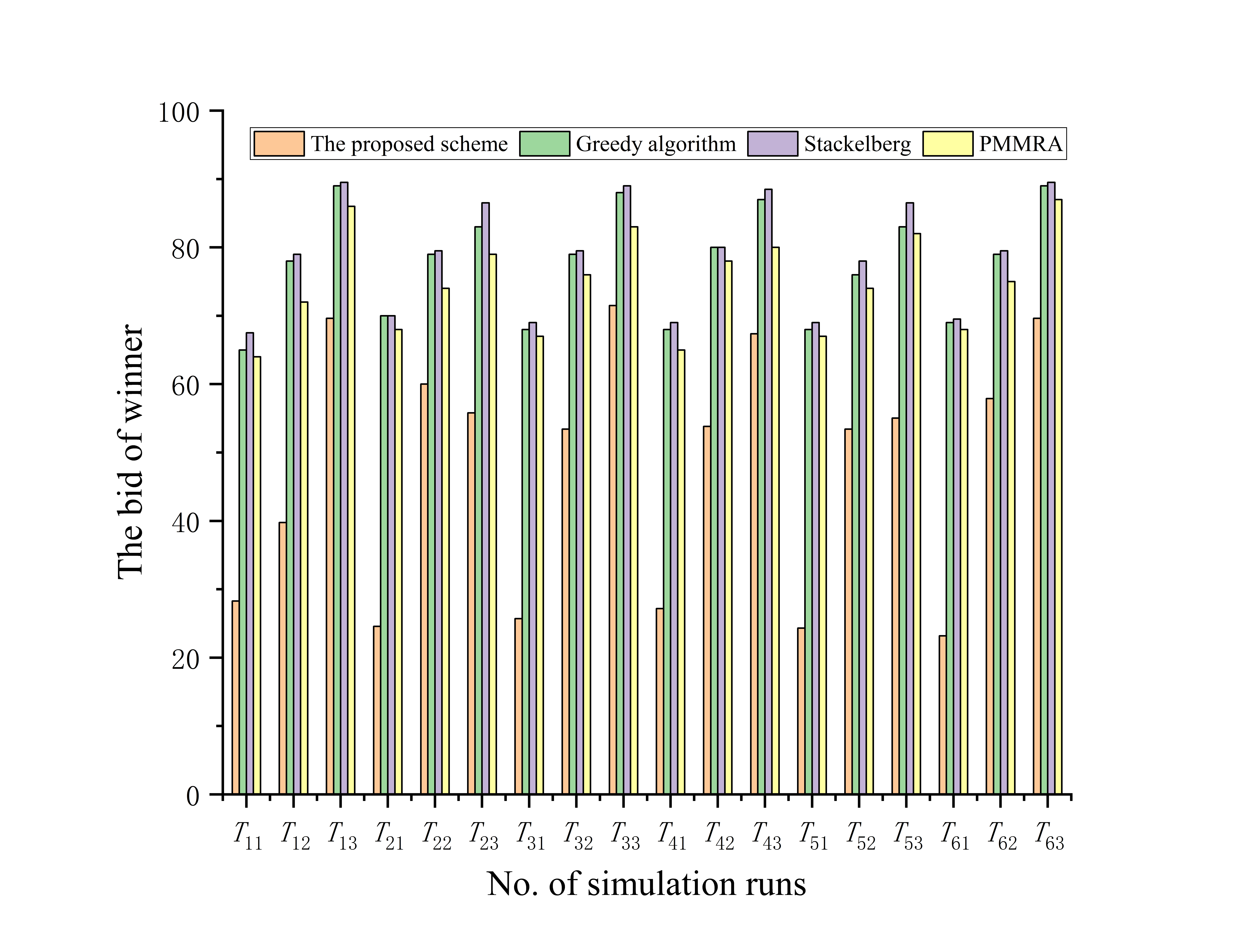}
\caption{Bids for different
schemes.\protect\footnotemark[4]}
\label{fig_6}
\end{figure}

\section{Conclusion}
In this letter, we propose a pricing scheme based on all-pay auction to stimulate EUs to offload tasks as well as maximize the total system profit.  Simulation results demonstrate the proposed scheme can make the EU with low bid get service at less than the value of the resource itself, and verify the effectiveness of the proposed scheme in terms of bid and total system profit.

\appendix
Denote that there are 3 sets, 3\textit{n} EUs, and the following two allocation schemes. Scheme 1: the number of EUs in each set  is \textit{n}. Scheme 2:  set the number of EUs in each set to be (\textit{n} - 1), \textit{n}, (\textit{n} + 1), respectively. Denote the total revenue (i.e., the sum of the bids of all EUs) of the two schemes as $P_1$ and $P_2$, respectively. In the following proof, we consider  \textit{A} with the same value, and the proof principle is similar when \textit{A} is different.
When all EUs valuations are the same, according to Eq.(\ref{eq9}), we  set ${\int_{v_{\min }^i}^vt d F^{n-1}(t)}$ = \textit{C}, then the following equation is obtained:
    \begin{equation}
    \begin{aligned}
    \label{eq24}
        P_2-P_1&=
        \frac{(n-1)C}{1-\frac{\lambda}{n-1}}+\frac{nC}{1-\frac{\lambda}{n}}+\frac{(n+1)C}{1-\frac{\lambda}{n+1}}-\frac{3nC}{1-\frac{\lambda}{n}}\\&=\frac{2C\lambda^2}{(n-1-\lambda)(n+1-\lambda)(n-\lambda)} \geq 0.
    \end{aligned}
    \end{equation}

If the EUs valuations are different, according to Eq.(\ref{eq22}), when \textit{A} increases,  $v_0$ will also increase. For $\textup{EU}_1$, $\textup{EU}_2$ and $\textup{EU}_3$, if $A_1$ $<$ $A_2$ $<$ $A_3$, the corresponding $v_0^1$ $<$ $v_0^2$ $<$ $v_0^3$ (or even $A_1$ $<$ $v_0^2$, $A_2$ $<$ $v_0^3$). According to Algorithm \ref{alg:alg1}, $\textup{EU}_1$ valuations are rarely  assigned in a set with $\textup{EU}_2$ valuations, and $\textup{EU}_2$ valuations are rarely  assigned in a set with $\textup{EU}_3$ valuations. That means  the set with low valuation and the set with high valuation has a lower number of EUs. According to Fig. \ref{fig_3}, $C_1 > C_2$, $C_3 > C_2$. At last, referring to Eq.(\ref{eq24}), we can obtain  that $P_1 < P_2$.

\ifCLASSOPTIONcaptionsoff
  \newpage
\fi

\footnotetext[4]{$T_{11}$ is the winner bid at $\textup{EC}_1 $ in the first trial, the rest can be speculated in the same manner.}


\begin{thebibliography}{1}

\bibitem{HK}
 X. Zhou, M. Bilal, R. Dou, J. J. P. C. Rodrigues, Q. Zhao, J. Dai and X. Xu, ``Edge Computation Offloading With Content Caching in 6G-Enabled IoV," {\em IEEE Trans. Intell. Transp. Syst.}, Early Access, 2023. 

\bibitem{XC}
 B. Li, Z. Qian, L. Liu, Y. Wu, D. Lan and C. Wu, ``Computation Offloading for Edge Computing in RIS-Assisted Symbiotic Radio Systems," {\em IEEE Trans. Netw. Sci. Eng.},  Early Access,  2023. 

\bibitem{WQ}
Q. Wang, S. Guo, J. Liu, C. Pan and L. Yang, ``Profit Maximization Incentive Mechanism for Resource Providers in Mobile Edge Computing,"  {\em IEEE Trans. Serv. Comput.}, vol. 15, no. 1, pp. 138-149,  Jan.-Feb. 2022.

\bibitem{M}
M. Liu and Y. Liu, ``Price-Based Distributed Offloading for Mobile-Edge Computing With Computation Capacity Constraints,"  {\em IEEE Wireless Commun. Lett.}, vol. 7, no. 3, pp. 420-423, Jun. 2018.



\bibitem{GQ}
B. Balzer, A. Rosato and J. von Wangenheim, ``Dutch vs. first-price auctions with expectations-based loss-averse bidders," {\em   J. Econ. Theory}, vol. 205, Art. no. 105545, Oct. 2022.

\bibitem{JF}
Y. Cherapanamjeri, C. Daskalakis, A. Ilyas, and M. Zampetakis, ``Estimation of standard auction models," {\em EC '22}, pp. 602-603, Jul. 2022.

\bibitem{SS}
S. Sareen, ``Posterior odds comparison of a symmetric low‐price, sealed‐bid auction within the common‐value and the independent‐private‐values paradigms." {\em JAE}, vol. 14, no. 6, pp. 651-657, Nov. 1999.

\bibitem{AA}
P. Choi, F. Munoz-Garcia, ``All-Pay Auctions and Auctions with Asymmetrically Informed Bidders." {\em STBE: Auction Theory}, pp 125-163, Feb. 2021.

\bibitem{UG}
U. Gneezy and R. Smorodinsky, ``All-pay auctions—an experimental study." {\em JEBO}, vol. 61, no. 2, pp. 255-275, Oct. 2006.


\end{thebibliography}
\end{document}